\documentclass[lettersize,journal]{IEEEtran}
\usepackage{amsmath,mathtools,amsfonts,amssymb}
\usepackage{amsthm}
\usepackage{hyperref}
\usepackage{xcolor}

\usepackage{algorithm}
\usepackage{algpseudocode}
\usepackage{array}
\usepackage[caption=false,font=normalsize,labelfont=sf,textfont=sf]{subfig}
\usepackage{textcomp}
\usepackage{stfloats}
\usepackage{url}
\usepackage{verbatim}
\usepackage{graphicx}
\usepackage{cuted}
\usepackage{float} 
\usepackage{multirow}
\usepackage{cancel}
\def\BibTeX{{\rm B\kern-.05em{\sc i\kern-.025em b}\kern-.08em
    T\kern-.1667em\lower.7ex\hbox{E}\kern-.125emX}}
\usepackage{balance}
\usepackage{booktabs} 
\usepackage{caption}
\newtheoremstyle{sltheorem}
{}                
{}                
{}        
{10pt}                
{\bfseries}       
{:}               
{ }               
{}                
\theoremstyle{sltheorem}
\newtheorem{theorem}{Theorem}
\newtheorem{definition}{Definition}
\newtheorem{corollary}{Corollary}

\usepackage{cite}
\begin{document}
\title{Reformulating Energy Storage Capacity Accreditation Problem with Marginal Reliability Impact}
\author{Qian Zhang,~\IEEEmembership{Graduate Student Member, IEEE,} Feng Zhao,~\IEEEmembership{Senior Member, IEEE,} Tongxin Zheng,~\IEEEmembership{Fellow, IEEE,} Le Xie,~\IEEEmembership{Fellow, IEEE}
\thanks{Qian Zhang and Le Xie are with John A. Paulson School of Engineering and Applied Sciences, Harvard University, Allston, MA, USA. Feng Zhao and Tongxin Zheng are with the ISO New England Inc., Holyoke, MA, USA. This work was initiated during Qian Zhang’s summer internship at ISO New England. (correspondence e-mail: qianzhang@g.harvard.edu).}}

\markboth{Journal of \LaTeX\ Class Files,~Vol.~18, No.~9, September~2020}%
{How to Use the IEEEtran \LaTeX \ Templates}

\maketitle

\begin{abstract}
To enhance the efficiency of capacity markets, many electricity markets in the U.S. are adopting or planning to implement marginal capacity accreditation reforms. This paper studies energy storage capacity accreditation using Marginal Reliability Impact (MRI), defined here as a directional derivative of expected unserved energy (EUE) with respect to the qualified-capacity path. We represent the commonly used reliability-dispatch rule for storage as a linear program, which enables direct MRI extraction from dual variables instead of repeated brute-force perturbation runs. The reformulation shows that the EUE value function is continuous and piecewise linear and that the storage MRI remains non-negative under the adopted reliability-dispatch adequacy model. We also report baseline validation against the perturbation benchmark and a focused sensitivity discussion on round-trip efficiency and time-step resolution. Numerical results from a modified California system confirm the dual-based MRI and provide practical guidance for qualified-capacity design and storage modeling in resource adequacy studies.
\end{abstract}

\begin{IEEEkeywords}
Capacity Market, Energy Storage, Capacity Accreditation, Marginal Reliability Impact
\end{IEEEkeywords}

\section{Introduction}
The Forward Capacity Market (FCM) ensures long-term resource adequacy in deregulated electricity markets by providing financial incentives for resources to commit their capacity in advance, thereby meeting reliability standards typically defined by a loss of load expectation (LOLE) of "one day in ten years." \cite{joskow2006competitive,cramton2013capacity}. However, the rapid integration of renewable energy and energy storage technologies has introduced new challenges to the conventional FCM framework due to the intermittent and variable nature of renewables and the unique operational characteristics of storage. In recent years, many Independent System Operators (ISOs) have been exploring or implementing new reforms targeting both the demand and supply sides of the FCM \cite{iso_ne_fcm_demand_curve,nyiso2016,pjm2022,miso_accreditation}. On the demand side, the capacity demand curve has been reconstructed from fixed capacity requirements or vertical lines to price-elastic demand curves based on the estimation of marginal reliability benefits of capacity and the value of loss load (VOLL) in FCM \cite{iso_ne_fcm_demand_curve,zhao2017constructing}. On the supply side, seasonal and locational factors have been incorporated into resource models to make the Resource Adequacy Assessment (RAA) process more precise, while capacity accreditation methods are being re-evaluated to appropriately account for the diverse contributions of different resources to system reliability \cite{pjm2022,zhao2020economic,miso_accreditation}. Additional reform suggestions from the academic community have also contributed by exploring innovative methods to address the complexities of an evolving resource mix \cite{bothwell2017crediting,wang2021crediting,zhang2023power,mays2023accreditation}. Together, these efforts aim to establish a framework where accredited capacity is more substitutable across resource types, thereby enhancing market efficiency and facilitating the transition to a decarbonized grid.

Among the various reform initiatives, marginal capacity accreditation has become a pivotal element in both demand and supply side reforms across most ISOs in the US. Unlike \emph{average} methods, which compensate each resource based on the aggregate reliability benefit of all units of that resource type, the \emph{marginal} approach compensates resources based on the incremental reliability benefit provided by the next unit of that type \cite{fengmri}. Reflecting these advantages, the Federal Energy Regulatory Commission (FERC) has approved marginal-based accreditation reforms for several ISOs in recent years \cite{ferc_nyiso,ferc_pjm}. In practice, two primary marginal accreditation metrics are widely used: marginal effective load carrying capability (ELCC) and MRI. Both are typically implemented via gradient-estimation procedures \cite{qiancredit}. This accreditation method is especially important for storage because its physical MW capability is not a homogeneous capacity product but depends on scarcity duration, charging opportunities, and the modeled resource mix.


Despite its advantages, the marginal accreditation methods present computational challenges as compared to class-level average accreditation, primarily due to the need to accurately measure the \emph{local reliability slope} of each resource's contribution to system reliability \cite{ISONE2022}. In current practice, this is usually done through perturbation analysis: the adequacy model is rerun after a small change in the accredited quantity, and the resulting finite difference is treated as MRI. This benchmark is computationally expensive and can also be sensitive to both perturbation step size and sample size \cite{qiancredit}. These issues are particularly pronounced for energy-limited resources, especially energy storage \cite{e3_2021,levin2023energy}.

With energy limited resources, load cannot be assumed met independently in each operating period. For energy storage, the ability to serve load in one time period is highly dependent on the operating conditions and energy availability from previous periods, complicating the marginal contribution calculation \cite{singh1988efficient,xu2013power,stephen2022impact}. The challenge is not only intertemporal but also structural: accreditation depends on both power and energy capability, so the marginal path must be stated carefully. Additionally, modeling appropriate charging and discharging orders of multiple storage resources remains a complex task \cite{evans2019minimizing,stephen2022impact}. These challenges not only increase the computational burden but also raise questions about the accuracy and fairness of marginal accreditation in accurately reflecting the reliability value of diverse resources.

In this paper, we address the energy storage capacity accreditation problem by reformulating it under the MRI framework. Rather than relying on perturbation results from commercial resource adequacy (RA) simulators, we represent the reliability-dispatch rule as a linear program, interpret MRI as a directional derivative, and then validate the resulting dual-based MRI against the perturbation benchmark. Connecting the storage dispatch rule to parametric optimization theory makes the local accreditation signal explicit and separates genuine modeling effects from numerical-step-size artifacts. The key contributions of this work are as follows: 
\begin{itemize} 
\item Enhancing the understanding of energy storage’s marginal reliability impact on system performance. 
\item Examining the influence of qualified capacity definitions and other factors on storage capacity accreditation, offering practical insights for system operators. 
\item Comparing storage capacity accreditation under different reliability criteria, offering valuable suggestion for policymakers in setting future reliability standards
\end{itemize}

The remainder of this paper is organized as follows: Section \ref{sec2} introduces the definition of MRI and its significance in the capacity market. Section \ref{sec3} provides an overview of different energy storage dispatch models in RA assessments, with a detailed explanation of the most commonly used reliability dispatch model. The main contributions of this paper are presented in Section \ref{sec4}, where the storage model is reformulated as an optimization problem to maximize system reliability. Section \ref{sec5} discusses key factors that may influence the reformulated model. Finally, Section \ref{sec6} demonstrates the efficiency of the proposed method through a modified California system, including comparative analysis.

\section{Marginal Reliability Impact} \label{sec2}
The MRI of a resource indicates how a small change in the resource's qualified capacity impacts system reliability, as measured by reliability metrics, such as the commonly used Expected Unserved Energy (EUE).
\subsection{MRI as a Rectifier for Qualified Capacity}

\begin{definition}[Marginal Reliability Impact \cite{fengmri}]
Let $q$ denote a one-dimensional accreditation path in the underlying storage-parameter space. We define the marginal reliability impact (MRI) as the local reduction in EUE induced by a marginal increase in Qualified Capacity (QC):
\begin{equation}
    \mathrm{MRI} = -\frac{\partial \mathrm{EUE}}{\partial \mathrm{QC}}.
\end{equation}
\end{definition}

Away from breakpoint values of the underlying parametric LP, the value function is affine in QC, and the MRI is an ordinary derivative. At breakpoint values, where the active constraint set changes, EUE remains piecewise linear but may fail to be differentiable. In such cases, we interpret MRI as the one-sided directional derivative along the accreditation path $q$, that is, the local slope of the relevant adjacent linear segment.

Throughout the paper, we focus on EUE, for which the mapping function is piecewise linear in QC and can fail to be differentiable only at turning points where a marginal change in capacity alters the set of shortage hours \cite{qiancredit}. In practice, derivatives are evaluated around the installed capacity requirement (ICR), which is determined by the system operator for the relevant delivery year. We therefore assume that the baseline evaluation point is \emph{regular}, meaning that it does not lie exactly at such a kink. This is a mild assumption in practice and is typically satisfied when the baseline system exhibits nonzero shortfall under the chosen reliability metric.

\begin{figure}[H]
\centering
\includegraphics[scale=0.75]{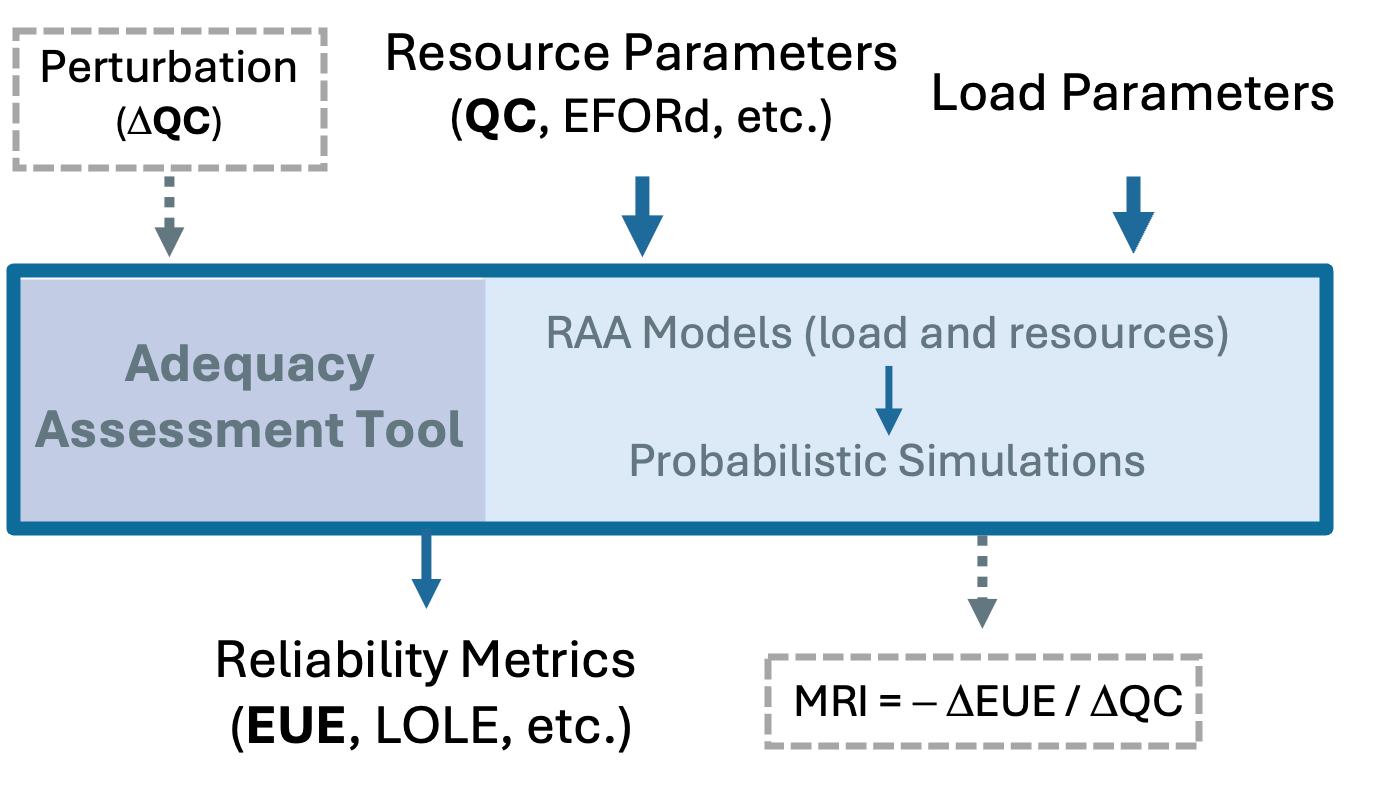}
  \caption{Sequential calculation process of reliability metrics and MRI}
  \label{mrical}
\end{figure}

The process for calculating reliability metrics and MRI is illustrated in Fig. \ref{mrical}, where the adequacy assessment tool is a probabilistic simulator utilizing the Monte Carlo method for resource adequacy analysis. The figure should be read as a sequence of repeated adequacy runs rather than as a single-step calculation: the baseline simulation first determines the reliability metric at the chosen ICR, after which either perturbation runs or the dual-based reformulation is used to recover the local MRI. To apply MRI for capacity accreditation, the MRI of an ideal resource is used as a reference.

\begin{definition}[relative Marginal Reliability Impact \cite{fengmri}] The relative Marginal Reliability Impact (rMRI) is defined as a ratio (unitless) of the resource's MRI to a perfect capacity's MRI:
\begin{equation}
    \text{rMRI} = \frac{\text{MRI}}{\text{MRI}_\text{perfect}}.
\end{equation}   
where the $\text{MRI}_\text{perfect}$ is typically defined as the MRI of a thermal generator resource with \emph{zero} outage rate. Other reference technologies are possible in principle, but the zero-outage thermal benchmark remains the standard normalization adopted in this paper. 
\end{definition}
After calculating the rMRI for each resource, the new credited capacity, i.e. Qualified MRI Capacity (QMRIC), will be calculated in MW:
\begin{equation}
    \text{QMRIC} = \text{QC} \times \text{rMRI}.
\end{equation}
The QC represents the \emph{physical} capacity of resources. For conventional generators this is typically a qualified MW capability close to nameplate power, while for renewable resources it is a derated accredited MW quantity based on modeled availability and performance. In contrast, the QMRIC reflects the \emph{market-accredited} quantity of a resource, which is directly used in the FCM clearing process.

\subsection{The Role of MRI in Forward Capacity Market}
The FCM functions as a market mechanism to ensure that power systems have adequate resources to meet future electricity demand. In ISO New England, system operators first determine the system-wide and zonal ICR based on expected system conditions three years prior to the operating period, which serves as the equilibrium point for MRI calculation. Subsequently, the QMRIC for the market is derived by discounting the QC with its rMRI. Using the estimated VOLL and resource bids, we can then construct the demand and supply curves for the FCM \cite{zhao2017constructing}. This process is illustrated in Fig. \ref{mrifcm}.
\begin{figure}[!t]
\centering
\includegraphics[scale=1]{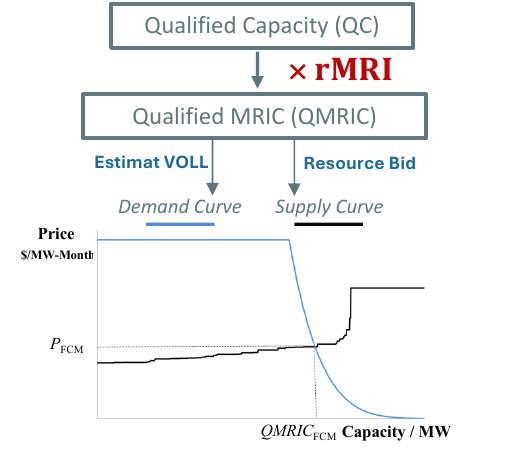}
  \caption{The Role of MRI in FCM}
  \label{mrifcm}
\end{figure}

Different from other \emph{average} reliability impact methods, the MRI-based approach provides more appropriate resource accreditation results and offers a clearer economic interpretation \cite{zhao2017constructing,zhao2020economic}. However, to ensure the efficiency of the FCM, certain properties of EUE and MRI are desired, particularly \emph{continuity} and \emph{non-negativity}, respectively. 

\section{Energy Storage Dispatch Model} \label{sec3}
In this section, we compare different storage dispatch models, with a focus on the \emph{most commonly} used model (reliability dispatch) for resource adequacy analysis.
\subsection{Overview of Energy Storage Dispatch Assumption}
Unlike thermal resources, the ability of energy storage to discharge in a given time period is highly dependent on the energy surplus from previous periods. As a result, the contribution of energy storage to system resource adequacy is determined primarily by the dispatch model rather than by the unit’s failure rate. Following the classification in Stephen et al. \cite{stephen2022impact}, dispatch models for multiple energy storage units can be classified into three categories. The first, \emph{expectation dispatch}, converts energy storage devices into non-chronological probabilistic models based on their historical behavior or average stochastic outcomes. Similar approaches are applied to intermittent resources and demand response, allowing system reliability to be evaluated using conventional methods \cite{Logan}.

However, under the expectation dispatch model, energy storage lacks adaptability to different operational scenarios, often resulting in an undervaluation of its resource adequacy contribution. The second category, \emph{economic dispatch} models, more accurately reflect energy storage behavior, but by the difficulty of modeling price-responsive and intertemporal user behavior, which requires strong and uncertain assumptions over long planning horizons.

To address the limitations of the previous models, \emph{reliability dispatch} models are used to quantify an adequacy-oriented benchmark, i.e., the reliability value of storage under the adopted RA dispatch rule rather than realized market behavior. These models are widely used by many ISOs for system resource adequacy analysis and storage capacity accreditation, which is the main focus of this paper.

\subsection{Reliability Dispatch}
\indent Assume each Monte Carlo experiment will sample a profile of available capacity during the observation period $\mathcal{T} := \{1,2,...,T\}$. The net power signal is used for storage charging decisions to capture the adaptability of energy storage under the reliability dispatch model, which is defined below. 
\begin{definition}[Net Power Signal]
Let $C_t(k)$ be a possible value of total available capacity except for energy storage during time interval $t$ for the $k$-th Monte Carlo experiment, where $N_p$ Monte Carlo experiment profiles are generated independently with equal probability contribution. The \emph{net power signal} is defined as the total available capacity except for energy storage minus the demand. The net power signal at time interval $t$ of the $k$-th Monte Carlo experiment can be expressed as:
\begin{equation}
P_t(k) = C_t(k) - L_t(k)    
\end{equation}
where $L_t(k)$ is the system demand level at time interval $t$ of $k$-th Monte Carlo experiment. Because the demand profiles are statistically independent, the final MRI used for capacity accreditation is computed as the average MRI across all profiles. In the following sections, we therefore restrict attention to a single profile, denoted by $P_t$.
\end{definition}

\indent Because resource adequacy is a long-term planning exercise, a key assumption of the energy storage dispatch model is that decision-makers possess only limited information about future system conditions. Consequently, rather than employing a multi-interval economic dispatch formulation, most dispatch models embedded in reliability assessment software adopt \emph{single-interval greedy algorithms}. Under this paradigm, storage devices charge from the grid whenever surplus generation is available ($P_t > 0$) and discharge only when required to mitigate a power deficit ($P_t < 0$). The rule is still chronological because the state of charge is propagated from one interval to the next; storage value is therefore path-dependent even though the dispatch decision is taken one interval at a time.

In most resource adequacy simulators, such as GE MARS and PRAS, storage that would be able to discharge the longest at their maximum rate is discharged first, and storage that would take the longest time to charge at their maximum charge rate is charged first \cite{GEMARS, stephen2021probabilistic}. This greedy dispatch algorithm is proven to minimize the expected unserved energy compared with any other storage dispatch methods \cite{evans2019minimizing}, and we will prove it in the next section.

To illustrate this dispatch model, we denote by $\mathcal{N}_s  \doteq \{D_1, D_2, ..., D_{N_s}\}$ the set of energy storage devices available in the system. Let $\overline{S}_i$ and $\overline{x}_i$ be the energy and power capacity of storage $D_i$, so that its state of charge (SoC) and power output during time interval $t$ satisfy $S_{i,t} \in [0, \overline{S}_i]$ and $x_{i,t} \in [-\overline{x}_i,\overline{x}_i]$, respectively. The initial SoC is predefined by system operator, typically at full charged state $S_{i,0} = \overline{S}_i$ \cite{isone_ge_mars_tech_session2_2025}.

\begin{algorithm}[H]
\caption{Greedy Algorithm in Energy Storage Reliability Dispatch (modified from \cite{evans2019minimizing})}\label{alg:dis}
\begin{algorithmic}[1]
\State \textbf{Input:} net power signal \(P_t\)
\State \textbf{Input:} initial remaining storage durations \(\{l_i(t-1)\}_{i \in \mathcal{N}_s}\)
\State \textbf{Input:} storage parameters \(\{\overline{x}_i,\overline{l}_i\}_{i \in \mathcal{N}_s}\)
\For{\(i \in \mathcal{N}_s\)}
    \State \(l_i'(t-1) = \max \{l_i(t-1) - 1, 0\}\)
    \State \(l_i''(t-1) = \min \{l_i(t-1) + 1, \overline{l}_i\}\)
\EndFor
\If{\(P_t < 0\)}
    \State Sort \(\{l_i(t-1), l_i'(t-1)\}_{i \in \mathcal{N}_s}\) in descending order
    \State Determine the reached remaining storage duration \(\underline{L}(t)\) from the discharge target \(-P_t\)
    \For{\(i \in \mathcal{N}_s\)}
        \State \(u_i(t) = \max\{\min\{l_i(t-1)-\underline{L}(t),1\},0\}\)
    \EndFor
\Else
    \State Sort \(\{l_i(t-1), l_i''(t-1)\}_{i \in \mathcal{N}_s}\) in ascending order
    \State Determine the reached remaining storage duration \(\overline{L}(t)\) from the charge target \(P_t\)
    \For{\(i \in \mathcal{N}_s\)}
        \State \(u_i(t) = -\max\{\min\{\overline{L}(t)-l_i(t-1),1,\overline{l}_i-l_i(t-1)\},0\}\)
    \EndFor
\EndIf
\State Compute the updated duration, SoC, and unserved energy:
\Statex \(l_i(t)=l_i(t-1)-u_i(t)\), \(S_{i,t}=\overline{x}_i l_i(t)\),
\Statex \(s(t)=\max\left\{-P_t-\sum_{i \in \mathcal{N}_s}\overline{x}_i u_i(t),0\right\}\)
\State \textbf{Output:} \(\{l_i(t),u_i(t)\}_{i \in \mathcal{N}_s}, s(t)\)
\end{algorithmic}
\end{algorithm}

\begin{definition}[Remaining Storage Duration] \label{def:remain_storage}
For storage device \(D_i\), the remaining storage duration is defined as the duration for which the device can continue discharging at its maximum power:
\begin{equation}
l_i(t)=\frac{S_{i,t}}{\overline{x}_i}.
\end{equation}
Similarly, the maximum storage duration of device \(D_i\) can be expressed as $\overline{l}_i(t)=\overline{S}_i / \overline{x}_i$.
\end{definition}

By Definition~\ref{def:remain_storage}, \(l_i(t)\) is the remaining storage duration of device \(i\). Algorithm~\ref{alg:dis} shows that, during a deficiency interval, reliability dispatch equalizes the post-dispatch remaining storage duration of the marginally discharged units subject to storage energy and power limits.

\section{MRI-based Storage Capacity Accreditation} \label{sec4}
For large-scale systems with energy storage resources, the reliability dispatch algorithms discussed above are integrated into Monte Carlo simulations for reliability analysis \cite{GEMARS}. As MRI calculations become increasingly important in supporting the FCM, we reformulate the MRI-based storage capacity accreditation problem to enhance both accuracy and computational efficiency. Our new reformulation can be interpreted as an analytical representation of this chronological reliability-dispatch rule, where uncertainty is captured by the Monte Carlo net-power trajectories.

\subsection{Reformulating Reliability Dispatch Problem}

For a fixed sampled net-power profile \(\{P_t\}_{t \in \mathcal{T}}\), let \(P_t^+ \doteq \max\{P_t,0\}\) and \(P_t^- \doteq \max\{-P_t,0\}\). We want to show that the Algorithm~\ref{alg:dis} used in storage reliability dispatch can achieve the optimal solution minimizing the unserved energy, which is formulated as $\mathcal{G}^\text{optimal}$:
\begin{equation} \label{esmodel2}
\begin{aligned}
\min_{S_{i,t},\,c_{i,t},\,d_{i,t},\,s(t)} \quad 
& \sum_{t=1}^{T} s(t) \\
\text{s.t.} \quad 
& \forall i \in \mathcal{N}_s,\quad \forall t \in \mathcal{T}: \\
& S_{i,0} = \text{constant}_i, \\
& S_{i,t} = S_{i,t-1} + c_{i,t} - d_{i,t},
&& (\phi_{it}) \\
& 0 \le S_{i,t} \le \overline{S}_i, 
&& (\mu_{it}^-, \mu_{it}^+) \\
& 0 \le c_{i,t} \le \overline{x}_i,
&& (\alpha_{it}^-, \alpha_{it}^+) \\
& 0 \le d_{i,t} \le \overline{x}_i,
&& (\beta_{it}^-, \beta_{it}^+) \\
& \sum_{i \in \mathcal{N}_s} c_{i,t} \le P_t^+,
&& (\kappa_t) \\
& P_t^- - \sum_{i \in \mathcal{N}_s} d_{i,t} \le s(t),
&& (\lambda_t) \\
& s(t) \ge 0,
&& (\gamma_t)
\end{aligned}
\end{equation}
Here \(c_{i,t}\) and \(d_{i,t}\) denote, respectively, the charged and discharged energy of storage device \(D_i\) over interval \(t\), measured in the same per-interval units as \(P_t\). After the duration change of variables \(S_{i,t}=\overline{x}_i l_i(t)\), \(c_{i,t}=\overline{x}_i c_i(t)\), and \(d_{i,t}=\overline{x}_i d_i(t)\), problem (\ref{esmodel2}) is converted exactly to the following full-horizon duration-domain formulation:
\begin{equation} \label{esmodel2t}
\begin{aligned}
\min_{l_i(t),\,c_i(t),\,d_i(t),\,s(t)} \quad
& \sum_{t=1}^{T} s(t) \\
\text{s.t.} \quad
& \forall i \in \mathcal{N}_s,\quad \forall t \in \mathcal{T}: \\
& l_i(t) = l_i(t-1) + c_i(t) - d_i(t), \\
& 0 \le l_i(t) \le \overline{l}_i, \\
& 0 \le c_i(t) \le 1, \\
& 0 \le d_i(t) \le 1, \\
& \sum_{i \in \mathcal{N}_s}\overline{x}_i c_i(t) \le P_t^+, \\
& P_t^- - \sum_{i \in \mathcal{N}_s}\overline{x}_i d_i(t) \le s(t), \\
& s(t) \ge 0.
\end{aligned}
\end{equation}
where \(c_i(t)\) and \(d_i(t)\) are the interval-wise charging and discharging durations, each bounded by one interval.

\begin{theorem}[Equivalence of Algorithm~\ref{alg:dis} and the Reformulated Optimal Dispatch Model $\mathcal{G}^\text{optimal}$] \label{thm:esequiv}
The unserved energy of storage reliability dispatch based on Algorithm~\ref{alg:dis} is no greater than that of \emph{any} other storage dispatch policy. In other words, it is an optimal solution within the optimal dispatch policy set \( \mathcal{G}^\text{optimal} \).
\end{theorem}

\begin{proof}
The proof follows the same primal-dual structure as \cite{qianramp}. The change of variables is exact, so it suffices to work with (\ref{esmodel2t}). Let \(\{l_i^*(t),u_i^*(t)\}_{i \in \mathcal{N}_s}\) and \(s^*(t)\) be the trajectory generated by Algorithm~\ref{alg:dis}. Define
\[
c_i^*(t)=\max\{-u_i^*(t),0\}, \qquad d_i^*(t)=\max\{u_i^*(t),0\}.
\]
We show that this greedy trajectory is primal feasible, construct a dual-feasible solution from it, and verify complementary slackness.

\subsubsection{Primal Feasibility}
\quad \\
\indent Let \(\underline{L}^*(t)\) and \(\overline{L}^*(t)\) denote the reached thresholds selected by Algorithm~\ref{alg:dis} on deficiency and surplus intervals, respectively. By the update formulas of Algorithm~\ref{alg:dis}, we have 
$0 \le d_i^*(t) \le \min\{1,l_i^*(t-1)\}$, and $
0 \le c_i^*(t) \le \min\{1,\overline{l}_i-l_i^*(t-1)\}$. Therefore, $l_i^*(t)=l_i^*(t-1)+c_i^*(t)-d_i^*(t)$, $0 \le l_i^*(t)\le \overline{l}_i$, and $0 \le c_i^*(t),d_i^*(t)\le 1$.
Moreover, the reached charging threshold is chosen so that
\[
\sum_{i \in \mathcal{N}_s}\overline{x}_i c_i^*(t)\le P_t^+.
\]
The algorithm defines $s^*(t)=\max\{P_t^- - \sum_{i \in \mathcal{N}_s}\overline{x}_i d_i^*(t),0\}$, which implies
\[
P_t^- - \sum_{i \in \mathcal{N}_s}\overline{x}_i d_i^*(t)\le s^*(t), \qquad s^*(t)\ge 0.
\]
Hence the greedy trajectory is primal feasible for (\ref{esmodel2t}).

\subsubsection{The Dual Problem}
\quad \\
\indent The dual of (\ref{esmodel2t}) is
\begin{subequations} \label{esdual}
\begin{align}
\max_{\phi,\mu,\alpha,\beta,\kappa,\lambda,\gamma} \quad
& - \sum_{i \in \mathcal{N}_s}\phi_{i1} l_i(0)
- \sum_{i \in \mathcal{N}_s}\sum_{t=1}^{T}\mu_{it}^+ \overline{l}_i \nonumber \\
 - &\sum_{i \in \mathcal{N}_s}\sum_{t=1}^{T}(\alpha_{it}^+ + \beta_{it}^+)
- \sum_{t=1}^{T}\kappa_t P_t^+ + \sum_{t=1}^{T}\lambda_t P_t^- \\
\text{s.t.} \quad
& \forall i \in \mathcal{N}_s,\ \forall t \in \mathcal{T}: \nonumber \\
& -\mu_{it}^- + \mu_{it}^+ + \phi_{it} - \phi_{i,t+1} = 0, \label{esdual_state} \\
& -\phi_{it} - \alpha_{it}^- + \alpha_{it}^+ + \kappa_t \overline{x}_i = 0, \label{esdual_charge} \\
& \phi_{it} - \beta_{it}^- + \beta_{it}^+ - \lambda_t \overline{x}_i = 0, \label{esdual_discharge} \\
& 1 - \lambda_t - \gamma_t = 0, \label{esdual_shortage} \\
& \mu_{it}^\pm,\alpha_{it}^\pm,\beta_{it}^\pm,\kappa_t,\lambda_t,\gamma_t \ge 0, \label{esdual_nonnegativity}
\end{align}
\end{subequations}
with boundary convention \(\phi_{i,T+1}=0\). The multiplier \(\phi_{it}\) is free in sign and represents the continuation value of one more unit of remaining storage duration carried from interval \(t\) to \(t+1\).

\subsubsection{Constructing the Dual Solution from Algorithm~\ref{alg:dis}}
\quad \\
\indent We construct the dual variables backward in time. For convenience define the scaled continuation values
\[
\theta_{it} \doteq \frac{\phi_{it}}{\overline{x}_i}.
\]
For \(t=1,\dots,T+1\), the recursive construction will maintain the invariant
\begin{equation} \label{esdualorder}
0 \le \theta_{it}\le 1, \qquad
l_i^*(t-1)\le l_j^*(t-1) \Rightarrow \theta_{it}\ge \theta_{jt}.
\end{equation}
That is, a shorter remaining storage duration carries a weakly larger continuation value per MW. The boundary condition \(\phi_{i,T+1}=0\) satisfies (\ref{esdualorder}) trivially.

Fix an interval \(t\), and suppose \(\{\phi_{i,t+1}\}_{i \in \mathcal{N}_s}\) has already been constructed so that (\ref{esdualorder}) holds for the state \(l_i^*(t)\). We now construct the interval-\(t\) dual multipliers directly from the greedy solution. Algorithm~\ref{alg:dis} implies the threshold relations
\[
\begin{aligned}
&\text{If} \quad P_t<0:\quad
d_i^*(t)=0 \Rightarrow l_i^*(t)\le \underline{L}^*(t), \\
&0<d_i^*(t)<1 \Rightarrow l_i^*(t)=\underline{L}^*(t), \quad d_i^*(t)=1 \Rightarrow l_i^*(t)\ge \underline{L}^*(t), \\
&\text{If} \quad P_t\ge 0:\quad
c_i^*(t)=0 \Rightarrow l_i^*(t)\ge \overline{L}^*(t),\\
& 0<c_i^*(t)<1 \Rightarrow l_i^*(t)=\overline{L}^*(t),\quad c_i^*(t)=1 \Rightarrow l_i^*(t)\le \overline{L}^*(t).
\end{aligned}
\]
The threshold-search step also has the exact-hit property: on a deficiency interval, either \(s^*(t)>0\) or \(\sum_i \overline{x}_i d_i^*(t)=P_t^-\); on a surplus interval, either all feasible charging headroom is exhausted or \(\sum_i \overline{x}_i c_i^*(t)=P_t^+\).

\textbf{Case 1: \(P_t<0\).} If \(s^*(t)>0\), set \(\lambda_t=1\). If \(s^*(t)=0\), choose \(\lambda_t \in [0,1]\) as the common scaled continuation value of units with \(l_i^*(t)=\underline{L}^*(t)\); if no unit ends exactly at \(\underline{L}^*(t)\), choose any value between the largest scaled continuation value on the set \(\{i:l_i^*(t)>\underline{L}^*(t)\}\) and the smallest scaled continuation value on the set \(\{i:l_i^*(t)<\underline{L}^*(t)\}\). Then $d_i^*(t)=0 \Rightarrow \theta_{i,t+1}\ge \lambda_t$, $0<d_i^*(t)<1 \Rightarrow \theta_{i,t+1}= \lambda_t$, and $d_i^*(t)=1 \Rightarrow \theta_{i,t+1}\le \lambda_t$.
Such a choice exists by the threshold relations above and (\ref{esdualorder}), using endpoint values when one side is empty. Then define
\[
\mu_{it}^+=0,\quad
\mu_{it}^-=
\begin{cases}
(\lambda_t-\theta_{i,t+1})\overline{x}_i, & l_i^*(t)=0,  \theta_{i,t+1}\le \lambda_t,\\
0, & \text{otherwise,}
\end{cases}
\]
\[
\phi_{it}=\phi_{i,t+1}+\mu_{it}^-,\qquad \theta_{it}=\frac{\phi_{it}}{\overline{x}_i},
\]
\[
\beta_{it}^-=
\begin{cases}
(\theta_{it}-\lambda_t)\overline{x}_i, & d_i^*(t)=0,\\
0, & \text{otherwise,}
\end{cases} \]
\[
\beta_{it}^+=
\begin{cases}
(\lambda_t-\theta_{it})\overline{x}_i, & d_i^*(t)=1,\\
0, & \text{otherwise.}
\end{cases}
\]
Finally set $\kappa_t=\max_{j \in \mathcal{N}_s}\theta_{jt}$, $\alpha_{it}^-=(\kappa_t-\theta_{it})\overline{x}_i$, $\alpha_{it}^+=0$, and $\gamma_t=1-\lambda_t$.

\textbf{Case 2: \(P_t\ge 0\).} If all feasible charging headroom is exhausted before using all surplus, set \(\kappa_t=0\). Otherwise choose \(\kappa_t \in [0,1]\) as the common scaled continuation value of units with \(l_i^*(t)=\overline{L}^*(t)\); if no unit ends exactly at \(\overline{L}^*(t)\), choose any value between the largest scaled continuation value on the set \(\{i:l_i^*(t)>\overline{L}^*(t)\}\) and the smallest scaled continuation value on the set \(\{i:l_i^*(t)<\overline{L}^*(t)\}\). Then $c_i^*(t)=0 \Rightarrow \theta_{i,t+1}\le \kappa_t$, $0<c_i^*(t)<1 \Rightarrow \theta_{i,t+1}= \kappa_t$, and $c_i^*(t)=1 \Rightarrow \theta_{i,t+1}\ge \kappa_t$.
Again such a choice exists by the threshold relations and (\ref{esdualorder}). Then define
\[
\mu_{it}^-=0,\qquad
\mu_{it}^+=
\begin{cases}
(\theta_{i,t+1}-\kappa_t)\overline{x}_i, & l_i^*(t)=\overline{l}_i,\theta_{i,t+1}\ge \kappa_t,\\
0, & \text{otherwise,}
\end{cases}
\]
\[
\phi_{it}=\phi_{i,t+1}-\mu_{it}^+,\qquad \theta_{it}=\frac{\phi_{it}}{\overline{x}_i},
\]
\[
\alpha_{it}^-=
\begin{cases}
(\kappa_t-\theta_{it})\overline{x}_i, & c_i^*(t)=0,\\
0, & \text{otherwise,}
\end{cases}\]
\[
\alpha_{it}^+=
\begin{cases}
(\theta_{it}-\kappa_t)\overline{x}_i, & c_i^*(t)=1,\\
0, & \text{otherwise.}
\end{cases}
\]
Finally set $\lambda_t=0$, $\gamma_t=1$, $\beta_{it}^-=\phi_{it}$, and $\beta_{it}^+=0$.

In both cases, \(\theta_{it}\) is either unchanged or clipped to the threshold value \(\lambda_t\) or \(\kappa_t\) on units whose post-dispatch state hits \(0\) or \(\overline{l}_i\). Since \(0 \le \theta_{i,t+1}\le 1\) and \(\lambda_t,\kappa_t \in [0,1]\), we also have \(0 \le \theta_{it}\le 1\). Moreover, on deficiency intervals $l_i^*(t)<\underline{L}^*(t) \Rightarrow \theta_{it}\ge \lambda_t$, $l_i^*(t)=\underline{L}^*(t) \Rightarrow \theta_{it}= \lambda_t$,$l_i^*(t)>\underline{L}^*(t) \Rightarrow \theta_{it}\le \lambda_t$, and on surplus intervals $l_i^*(t)<\overline{L}^*(t) \Rightarrow \theta_{it}\ge \kappa_t$, $l_i^*(t)=\overline{L}^*(t) \Rightarrow \theta_{it}= \kappa_t$, $l_i^*(t)>\overline{L}^*(t) \Rightarrow \theta_{it}\le \kappa_t$.
Hence the order relation in (\ref{esdualorder}) is preserved from \(t+1\) to \(t\), so the backward recursion is well defined from \(t=T\) to \(t=1\).

\begin{figure*}[t]
\vspace{-1.2em}
\begin{equation} \label{lag}
\begin{aligned}
&L(\boldsymbol{S},\boldsymbol{c},\boldsymbol{d},\boldsymbol{s},\boldsymbol{\phi},
\boldsymbol{\mu},\boldsymbol{\alpha},\boldsymbol{\beta},\boldsymbol{\kappa},
\boldsymbol{\lambda},\boldsymbol{\overline{S}},\boldsymbol{\overline{x}})
= \\
& V_2(\boldsymbol{s})
+ \sum_{i=1}^{N_s}\sum_{t=1}^{T}\phi_{it}(S_{i,t}-S_{i,t-1}-c_{i,t}+d_{i,t})
- \sum_{i=1}^{N_s}\sum_{t=1}^{T} \mu_{it}^- S_{i,t}
+ \sum_{i=1}^{N_s}\sum_{t=1}^{T} \mu_{it}^+ (S_{i,t}-\overline{S}_i) \\
&- \sum_{i=1}^{N_s}\sum_{t=1}^{T} \alpha_{it}^- c_{i,t}
+ \sum_{i=1}^{N_s}\sum_{t=1}^{T} \alpha_{it}^+ (c_{i,t}-\overline{x}_i)
- \sum_{i=1}^{N_s}\sum_{t=1}^{T} \beta_{it}^- d_{i,t}
+ \sum_{i=1}^{N_s}\sum_{t=1}^{T} \beta_{it}^+ (d_{i,t}-\overline{x}_i) \\
&+ \sum_{t=1}^{T}\kappa_t \Bigl(\sum_{i=1}^{N_s}c_{i,t} - P_t^+\Bigr)
+ \sum_{t=1}^{T}\lambda_t \Bigl(P_t^- - \sum_{i=1}^{N_s}d_{i,t} - s(t)\Bigr)
- \sum_{t=1}^{T}\gamma_t s(t)
\end{aligned}
\end{equation}
\hrulefill
\end{figure*}

\subsubsection{Verification of Dual Feasibility and Optimality}
\quad \\
\indent By construction, all dual variables except \(\phi_{it}\) are nonnegative. Equation (\ref{esdual_state}) holds because \(\phi_{it}=\phi_{i,t+1}+\mu_{it}^- - \mu_{it}^+\). Equation (\ref{esdual_charge}) holds by the definitions of \(\alpha_{it}^\pm\) and \(\kappa_t\), and (\ref{esdual_discharge}) holds by the definitions of \(\beta_{it}^\pm\) and \(\lambda_t\). Equation (\ref{esdual_shortage}) holds because \(\gamma_t=1-\lambda_t\). Therefore the constructed dual point is dual feasible.

Complementary slackness also follows directly from the construction. A positive \(\mu_{it}^-\) is assigned only when \(l_i^*(t)=0\), a positive \(\mu_{it}^+\) only when \(l_i^*(t)=\overline{l}_i\), a positive \(\alpha_{it}^+\) only when \(c_i^*(t)=1\), a positive \(\beta_{it}^+\) only when \(d_i^*(t)=1\), a positive \(\alpha_{it}^-\) only when \(c_i^*(t)=0\), and a positive \(\beta_{it}^-\) only when \(d_i^*(t)=0\). Moreover, $s^*(t)\gamma_t=0$, because \(s^*(t)>0\) implies \(\lambda_t=1\) and hence \(\gamma_t=0\). For the aggregate interval constraints,
\[
\lambda_t\Bigl(s^*(t)-P_t^-+\sum_{i \in \mathcal{N}_s}\overline{x}_i d_i^*(t)\Bigr)=0,
\]
because \(s^*(t)\) is defined from the shortage balance and, when \(s^*(t)=0\), the exact-hit property gives \(\sum_i \overline{x}_i d_i^*(t)=P_t^-\) whenever \(\lambda_t>0\). Likewise,
\[
\kappa_t\Bigl(P_t^+-\sum_{i \in \mathcal{N}_s}\overline{x}_i c_i^*(t)\Bigr)=0,
\]
because either the charging target is hit exactly or all feasible charging headroom is exhausted, in which case we set \(\kappa_t=0\).

We have therefore constructed a primal-feasible greedy trajectory and a dual-feasible multiplier vector satisfying complementary slackness for (\ref{esmodel2t}). By strong duality of linear programming, the greedy trajectory generated by Algorithm~\ref{alg:dis} is an optimal solution of (\ref{esmodel2}) and (\ref{esmodel2t}).
\end{proof}

\subsection{Storage Capacity Accreditation with MRI}

Theorem~\ref{thm:esequiv} shows that Algorithm~\ref{alg:dis} and the full-horizon LP (\ref{esmodel2}) are equivalent on every sampled trajectory. This equivalence allows MRI to be recovered directly from dual variables of (\ref{esmodel2}) rather than from repeated perturbation studies. Let \(V_2 := \sum_{t=1}^{T} s(t)\) denote the objective function value of (\ref{esmodel2}). The Lagrangian of (\ref{esmodel2}), parameterized by \(\overline{S}_i\) and \(\overline{x}_i\), is shown in \eqref{lag}.

By the \emph{Envelope Theorem}, the derivative of \(V_2\) with respect to the storage parameters can be read directly from the optimal dual multipliers:

\begin{equation}
\begin{aligned}
&\frac{\partial V_2}{\partial \overline{S}_i} = -\sum_{t=1}^{T} \mu_{it}^+\\
&\frac{\partial V_2}{\partial \overline{x}_i} = -\sum_{t=1}^{T} (\alpha_{it}^+ +\beta_{it}^+)
\end{aligned}
\end{equation}

\begin{theorem}[Non-negative MRI and Piecewise linear EUE: Multiple Energy Storage] \label{multithe}
If the QC for storage device $D_i$ is defined as either storage energy capacity $\overline{S}_i$ or power capacity $\overline{x}_i$ (or any positive linear combination thereof), the MRI of each storage is non-negative and the EUE is piecewise linear (both with respect to QC), regardless of system condition.
\end{theorem}
\begin{proof}
\indent Problem (\ref{esmodel2}) is a parametric linear program, and Theorem~\ref{thm:esequiv} shows that it is equivalent to the full chronological dispatch generated by Algorithm~\ref{alg:dis}. Increasing any individual \(\overline{S}_i\) or \(\overline{x}_i\) only relaxes the feasible set while preserving feasibility, so the optimal EUE cannot increase. Therefore,
\begin{equation}
\text{MRI}_{\overline{S}_i} = -\frac{\partial \mathrm{EUE}}{\partial \overline{S}_i}
= -\frac{\partial V_2}{\partial \overline{S}_i}
= \sum_{t=1}^{T}\mu_{it}^+ \ge 0,
\end{equation}
and the same monotonicity argument applies to \(\overline{x}_i\) because
\begin{equation}
\text{MRI}_{\overline{x}_i} = -\frac{\partial \mathrm{EUE}}{\partial \overline{x}_i}
= \sum_{t=1}^{T}(\alpha_{it}^+ + \beta_{it}^+) \ge 0.
\end{equation}

To establish the piecewise-linear property, note that after introducing the auxiliary variables \(s(t)\), problem (\ref{esmodel2}) is a parametric linear program in which \(\overline{S}_i\) and \(\overline{x}_i\) enter linearly through the right-hand sides of the bounds. Standard parametric-LP results imply that the optimal value is continuous and piecewise affine on a finite partition of the parameter space \cite{bertsimas1997introduction}. Hence the corresponding EUE is continuous and piecewise linear with respect to each QC definition.

\indent At breakpoints, the derivative need not be unique. In those cases, the optimal dual multipliers still provide valid directional derivatives or LP subgradients, which is the interpretation of MRI adopted in this paper.
\end{proof}

\begin{corollary}[Direct Calculation of MRI] 
\label{coro2}
If the QC for storage device $D_i$ is defined as either storage energy capacity $\overline{S}_i$ or power capacity $\overline{x}_i$, the MRI of each storage can be directly calculated based on the value of Lagrangian multipliers in (\ref{esmodel2}).
\end{corollary}

\emph{Remark}: After averaging over multiple Monte Carlo trajectories, the system-level MRI is the average of the scenario-level directional derivatives or subgradients. The non-negative property therefore remains valid, and the perturbation study in Section~\ref{sec6} is used to verify the dual MRI numerically.

\section{Discussion} \label{sec5}
\subsection{Qualified Capacity}
Under the MRI capacity accreditation scheme, the definition of QC is relatively flexible. For instance, defining QC as the storage's full power capacity yields the same accreditation result as defining it as half of the storage's power capacity.
\begin{theorem}[Invariant QMRIC] \label{sameqmric}
Let $QC_1$ and $QC_2$ represent the capacity qualification values derived from two different capacity qualification methods, and different qualifications do not change the dispatch models of storage resources in RAA. Suppose there exists a constant $\alpha \neq 0$ such that $QC_1 = \alpha QC_2$ for all different storage capacities. After applying marginal reliability impact correction, the resulting qualified MRI-based capacity values $QMRIC_1$ and $QMRIC_2$ are identical, i.e. $QMRIC_1 = QMRIC_2$.
\end{theorem}
\begin{proof}
Based on the chain rule for derivatives, we have
\begin{equation}
\begin{aligned}
QMRIC_1 = &QC_1\frac{\text{MRI}}{\text{MRI}_p} = -\frac{QC_1}{\text{MRI}_p}\frac{\partial \text{EUE}}{\partial QC_1} \\
 = &-\frac{\alpha QC_2}{\text{MRI}_p}\frac{\partial \text{EUE}}{\partial \alpha QC_2} = QMRIC_2
\end{aligned}
\end{equation}
\end{proof}
Theorem \ref{sameqmric} also holds for other resources, such as thermal and renewable generators. However, for energy storage or energy-limited resources, the energy capacity plays an important part in capacity qualification. If the QC is constituted by both storage power capacity $\Bar{x}$  and energy capacity $\Bar{S}$, their contribution factor plays an important role in QMRIC calculation. Supposing $QC = \beta_1 \Bar{x} + \beta_2 \Bar{S}$, then the QMRIC can be expressed by chain rule as:
\begin{equation}
\begin{aligned}
QMRIC = & QC\frac{\text{MRI}}{\text{MRI}_p} = -\frac{QC}{\text{MRI}_p}\frac{\partial \text{EUE}}{\partial QC} \\
=& -\frac{QC}{\text{MRI}_p}\frac{\partial \text{EUE}}{\partial (\beta_1 \Bar{x} + \beta_2 \Bar{S})}
\end{aligned}
\end{equation}
To make the above expression rigorous, the perturbation path must be specified explicitly. In this paper, a scalar QC perturbation is interpreted as a directional derivative in the \((\Bar{x},\Bar{S})\) plane satisfying \(\beta_1 \Delta \Bar{x} + \beta_2 \Delta \Bar{S} = 1\). In other words, the MRI depends on both the QC rule and the path along which power and energy are changed. Noticing that the unit of $\beta_2$ is 1/hour, how to decide the value of these contribution factors is still a big challenge and unsolved question. A practical heuristic is to use shortage-event statistics: if scarcity is dominated by sharp one-hour events, the operator should place more weight on \(\beta_1\); if multi-hour scarcity blocks dominate, more weight should be assigned to \(\beta_2\). Currently, most US ISOs only use the storage power capacity $\Bar{x}$ to calculate QC.
\subsection{Initial SoC}
In (\ref{esmodel2}), the initial state of charge (SoC) of each storage device is treated as a constant value, independent of the corresponding storage capacity. If the initial SoC of storage device \(D_i\) is assumed to equal its energy capacity, i.e., \(S_{i,0}=\overline{S}_i\), then \(\overline{S}_i\) enters not only the upper energy-capacity bounds but also the initial-condition term and the first-interval state-transition relation. Consequently, the derivative \(\partial V_2/\partial \overline{S}_i\) acquires additional first-interval dual terms.

The exact expression therefore depends on the sign of the first-interval net power signal \(P_1\), but the qualitative conclusion is unchanged: increasing \(\overline{S}_i\) still relaxes the feasible set and cannot increase EUE. Hence the non-negative MRI property established in Theorem~\ref{multithe} remains valid when the initial SoC is tied to energy capacity.

Different from SoC-dependent bidding in the energy market. In a multi-year forward capacity market study with many simulated 8760-hour trajectories, there is no unique physical ``first hour'' from which the starting SoC can be observed. Consistent with industry practice \cite{isone_ge_mars_tech_session2_2025}, we treat all storage initial-SoC as full and leave SoC-dependent capacity accreditation as a future market-design question.

\subsection{Round-trip Efficiency}
The round-trip efficiency (RTE) of an energy storage system is defined as the ratio of total energy output to total energy input. While various approaches exist for modeling RTE in resource adequacy analysis \cite{stephen2022impact,evans2019minimizing} or capacity accreditation \cite{wang2021crediting,qi2024capacity}, the impact of RTE on MRI has received relatively little attention.

In practice, it is common to decompose RTE into charging and discharging efficiencies. For many storage technologies used in RA studies, the charging efficiency is typically around \(85\%\), while the discharging efficiency is often close to 1. Therefore, similar to the commercial software GE MARS \cite{GEMARS}, we model efficiency loss only on the charging side and interpret \(\eta_i\) below as the effective charging efficiency of storage device \(D_i\). This simplification preserves the linear structure of the dispatch model while capturing the dominant energy-loss channel.

Let \(x_{i,t}\) denote discharge power delivered to the system and let \(y_{i,t}\) denote charging power drawn from the system. Then the charging/discharging part of (\ref{esmodel2}) is replaced by:

\begin{equation} \label{multirte}
\begin{aligned}
& 0 \le x_{i,t} \le \overline{x}_i,\qquad 0 \le y_{i,t} \le \overline{x}_i, \\
& S_{i,t} = S_{i,t-1} + \eta_i y_{i,t} - x_{i,t}, \\
& \sum_{i=1}^{N_s} y_{i,t} \le P_t^+, \\
& P_t^- - \sum_{i=1}^{N_s} x_{i,t} \le s(t), \\
& s(t) \ge 0
\end{aligned}
\end{equation}
\indent Together with the unchanged energy-capacity bounds \(0 \le S_{i,t} \le \overline{S}_i\), (\ref{multirte}) captures charging-side efficiency loss under the same \(P_t\)-based notation. In Algorithm~\ref{alg:dis}, the discharging branch is unchanged because the discharging efficiency is approximated by 1, whereas the charging branch must replace the one-interval candidate level by
\[
l_i''(t-1)=\min\{l_i(t-1)+\eta_i,\overline{l}_i\},
\]
so the charging step becomes
\[
u_i(t)=-\max\{\min\{\overline{L}(t)-l_i(t-1),\eta_i,\overline{l}_i-l_i(t-1)\},0\}.
\]
Equivalently, the full-horizon optimal dispatch model is modified by replacing the charging part of (\ref{esmodel2}) with (\ref{multirte}). The dual-variable accreditation method remains workable because the resulting problem is still a parametric linear program in \((\overline{S}_i,\overline{x}_i)\). Hence the optimal dual multipliers of the modified LP still provide valid directional derivatives or subgradients for MRI-based accreditation, although their numerical values will generally change with \(\eta_i\).

\subsection{Scope, Limitations, Market Behavior, and Network Extension}
The analytical reformulation should be read as an implementation of a marginal accreditation rule \emph{within} a given resource-adequacy model, not as a replacement for probabilistic RAA. Its accuracy still depends on the modeled distributions of load, renewable output, forced outages, weather correlations, transmission availability, fuel availability, storage duration, efficiency, initial SoC, and the capacity-market reference point. The storage reliability dispatch model used in this paper is the only widely used storage dispatch model in industry software. A market-oriented accreditation study would additionally require explicit assumptions on prices, bidding strategies, imperfect foresight, and performance penalties. Those assumptions are important, but they are market-design choices rather than intrinsic properties of storage reliability. Similarly to the social-welfare benchmark for storage in \cite{sioshansi2014energy}, the present model should be interpreted as an adequacy benchmark rather than as a prediction of realized merchant storage behavior.

The dual construction is nevertheless not tied to one dispatch philosophy. If an ISO adopts an alternative dispatch assumption and that rule can be represented as an optimization problem, the corresponding dual variables can be used to recover a dispatch-consistent MRI. This includes extensions in which storage must reserve part of its charging or discharging capability for ancillary services: such requirements can be modeled as reduced effective headroom or additional reserve-holding constraints, and the same dual logic continues to apply as long as the resulting adequacy benchmark remains an optimization problem. The same parametric formulation can also evaluate candidate storage resources that are not yet deployed by specifying their hypothetical power and energy capacities before they enter the fleet.

The computational scaling advantage also becomes clearer at larger system size. The perturbation benchmark scales with the number of resources times the number of Monte Carlo scenarios, whereas the dual method requires one solve per scenario for the entire storage fleet. Scenario decomposition is therefore embarrassingly parallel. If a DC network model or regional power transfer limit is added, the same logic extends to locational MRI: congestion changes the feasible dispatch region and the relevant dual variables become nodal reliability shadow prices.

\section{Case Study} \label{sec6}
\subsection{System Description}
The modified 2019 California system is used as the test case, with input data primarily sourced from the U.S. Energy Information Administration (EIA) \cite{eia860}. The open-source Python package \emph{assetra} is used to generate the net capacity surplus profile based on Monte Carlo resource adequacy simulations. Further details on the data and simulation settings can be found in \cite{assetra}.

In the absence of energy storage, net capacity surplus profiles for the system are generated from 100 independent simulations. The results for the summer peak month are shown as light grey curves in Fig. \ref{netcap}.
\begin{figure}[H]
\centering
\includegraphics[scale=0.5]{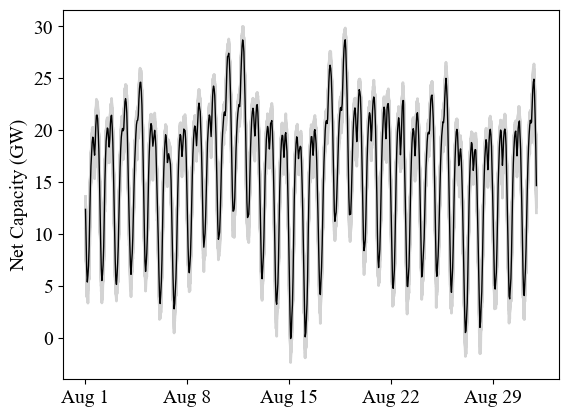}
  \caption{The system net capacity surplus without considering energy storage (August)}
  \label{netcap}
\end{figure}
Because California is a summer peak system, most energy-deficient hours are concentrated in July and August. The black line represents the mean value of the net capacity surplus, where the expected unserved energy without energy storage for the whole year is 1710.8 MWh.

\subsection{Energy Storage Capacity Accreditation}
Based on EIA data, about 1500 MW storage with 1800 MWh energy capacity was operated in California during 2019 \cite{eia860}. These energy storage devices are operated by 118 different utility companies, while over 50 \% of the storage capacity is less than 3 MW. For simplicity but to keep the main character of the storage, we aggregate the storage into 4 groups based on their duration, which is illustrated in Table \ref{groupcap}.
\begin{table}[!t] 
\caption{The capacity of aggregated energy storage groups}
\label{groupcap}
\centering
\begin{tabular}{cccc} 
\toprule
\multirow{2}{*}{\begin{tabular}[c]{@{}c@{}}\textbf{Group}\\  \# \end{tabular}}
& \multirow{2}{*}{\begin{tabular}[c]{@{}c@{}}\textbf{Power Capacity}\\  Aggregated (MW)\end{tabular}}
&\multirow{2}{*}{\begin{tabular}[c]{@{}c@{}}\textbf{Energy Capacity}\\  Aggregated (MWh)\end{tabular}}
&\multirow{2}{*}{\begin{tabular}[c]{@{}c@{}}\textbf{Duration Range}\\  (Hour)\end{tabular}}  \\
\\ \midrule
1&                 1127.6 &                      229.3 &              (0, 1] \\
2&                  284.4 &                      316.9 &              (1, 2] \\
3&                   82.5 &                      188.8 &              (2, 4] \\
4&                  84.1 &                     1084.2 &       $>$ 4 \\
\bottomrule
\end{tabular}
\end{table}
Under the reliability dispatch rule, the energy storage is dispatched to minimize the system's unserved energy. Supposing the storage's initial SoC is equal to its maximum energy capacity and no round trip loss, the value of EUE for the whole year after considering energy storage is reduced to 481.2 MWh.

Different from the \emph{average} reliability impact methods \cite{nyiso2016,pjm2022}, where the average EUE (or other reliability metrics) contribution is used for capacity accreditation, the MRI framework will calculate the \emph{marginal} reduction of EUE by each energy storage resource. Because the resource adequacy analysis is based on the simulation of complex stochastic discrete event systems, the MRI is typically calculated by the tedious perturbation method \cite{ISONE2024}. For example, the engineer will monitor the changing of EUE by adding 1 MW power capacity or 1 MWh energy capacity of the original storage resource as its marginal contribution to system reliability. Table \ref{MRIperturb} compares the results of this perturbation analysis of different aggregated storage resources.
\begin{table}[!t] 
\caption{The MRI of aggregated energy storage through perturbation analysis}
\label{MRIperturb}
\centering
\begin{tabular}{ccc} 
\toprule
\multirow{2}{*}{\begin{tabular}[c]{@{}c@{}}\textbf{Group}\\  \# \end{tabular}}
& \multirow{2}{*}{\begin{tabular}[c]{@{}c@{}}\textbf{MRI} (hour/year)\\  Power Capacity \end{tabular}}
&\multirow{2}{*}{\begin{tabular}[c]{@{}c@{}}\textbf{MRI} (1/year)\\  Energy Capacity \end{tabular}} \\ 
\\ \midrule
1&                 0 &                      1.02 \\
2&                  0.81 &                      0.21 \\
3&                   1.21 &                      0.01 \\
4&                  1.24 &                     0\\
\bottomrule
\end{tabular}
\end{table}

For shorter-duration storage, increasing energy capacity has a greater impact on system reliability, whereas for longer-duration storage, power capacity has a higher MRI. This distinction leads to varying rMRI values across different storage types.

\textbf{\emph{Remark}}: It is important to note that storage with a duration of less than one hour \emph{always} has an MRI of 0 when QC is defined by its power capacity. This occurs because the time step in RA simulations is typically set to 1 hour, preventing the capture of \emph{intra-hour} behavior of some resources \cite{Stephen2024}. This highlights the need for either higher-resolution RA simulations or the inclusion of energy capacity in QC definitions to enable more accurate capacity accreditation in the future.

\begin{figure}[H]
\centering
\includegraphics[scale=0.335]{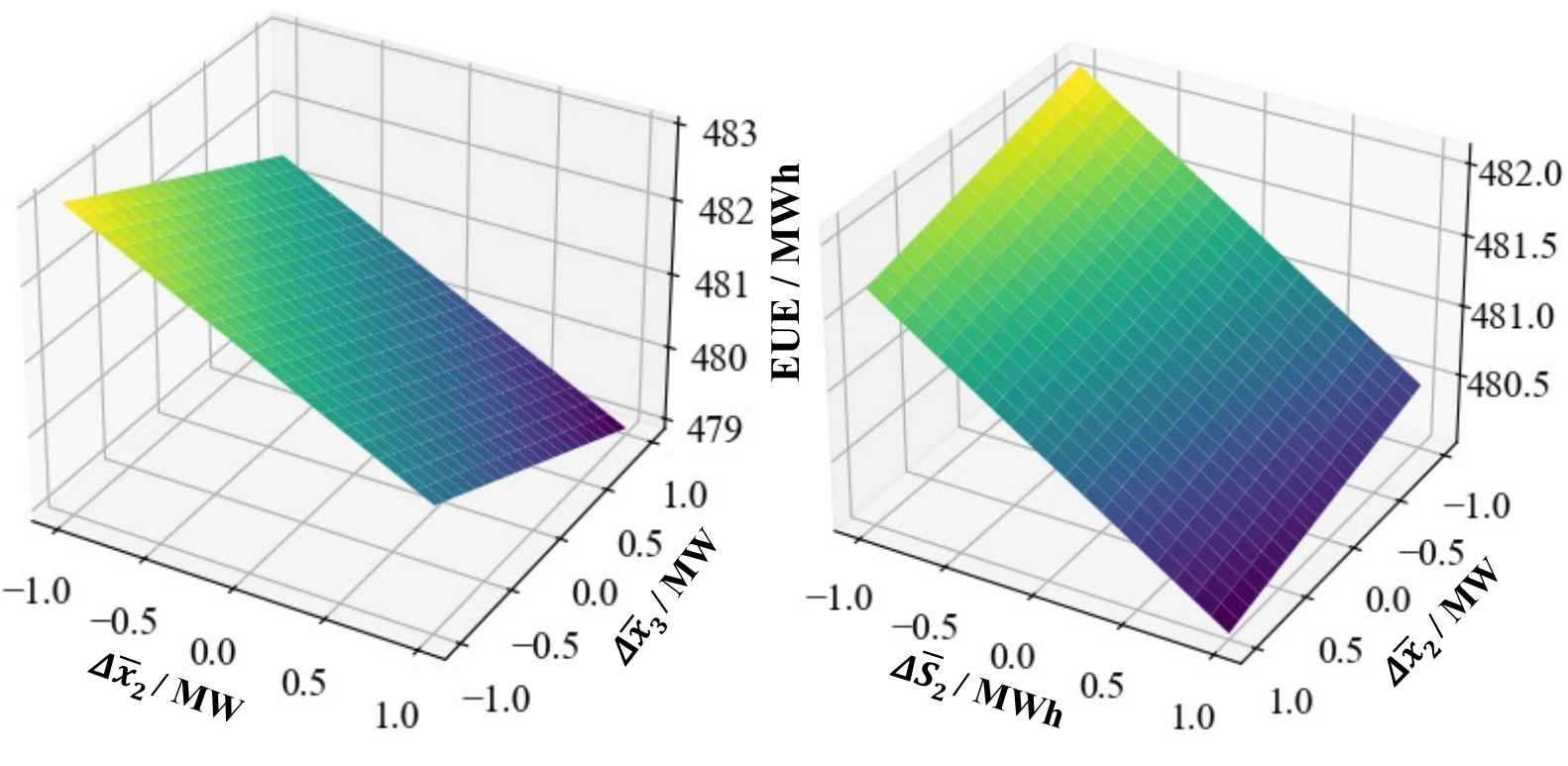}
  \caption{The value of EUE under the perturbation of different power and energy capacities}
  \label{perts}
\end{figure}

To further validate the correctness of Theorem \ref{multithe}, we continuously perturb the system with smaller step sizes for storage capacities, specifically 0.1 MW or 0.1 MWh. Fig. \ref{perts} left presents the EUE simulation results when simultaneously varying the power capacity of two storage units, while Fig. \ref{perts} right shows the results when both the energy and power capacities of the same storage unit are perturbed together.

Ideally, the EUE function is piecewise linear with respect to the QC of storage, but the interval length of each segment is typically much larger than the storage capacity. Consequently, the MRI values presented in Table \ref{MRIperturb} are the \emph{same} as the slope of the EUE function, as illustrated in Fig. \ref{perts}.

We additionally validate the dual MRI on the same frozen 100-scenario Monte Carlo sample used for the main California benchmark. Unless noted otherwise, these validation runs use a 1-hour simulation step, a fully charged initial state, round-trip efficiency of 1.0, and perturbation steps of $\Delta \Bar{x}=1$ MW and $\Delta \Bar{S}=1$ MWh.

Table \ref{MRIaccuracy} shows that the dual-based MRI reproduces the perturbation MRI to machine precision for both power and energy dimensions, while requiring only one LP solve per scenario instead of one perturbation run per resource and per capacity dimension. Because the same 100-scenario sample is used throughout, the values in Table \ref{MRIaccuracy} coincide with the original perturbation values in Table \ref{MRIperturb}.
\begin{table}[H]
\caption{Dual-versus-perturbation MRI on the frozen 100-scenario benchmark}
\label{MRIaccuracy}
\centering
\begin{tabular}{ccccc}
\toprule
\textbf{Group} & \textbf{Power MRI} & \textbf{Power MRI} & \textbf{Energy MRI} & \textbf{Energy MRI} \\
\# & \textbf{(Dual)} & \textbf{(Pert.)} & \textbf{(Dual)} & \textbf{(Pert.)} \\
\midrule
1 & 0.00 & 0.00 & 1.02 & 1.02 \\
2 & 0.81 & 0.81 & 0.21 & 0.21 \\
3 & 1.21 & 1.21 & 0.01 & 0.01 \\
4 & 1.24 & 1.24 & 0.00 & 0.00 \\
\bottomrule
\end{tabular}%
\end{table}

Table \ref{MRItime} compares the MRI solving time for all storage resources\footnote{All the problems are solved using 18 GB memory on the Apple M3 Pro. The mathematical models were formulated using gurobipy on Python and solved using Gurobi v11.0.}. While the perturbation-based method needs to solve the EUE of each capacity-changing scenario for different storage, our method only requires solving one optimization problem per scenario for the entire storage fleet.
\begin{table}[H]
\caption{Runtime comparison on the frozen 100-scenario benchmark}
\label{MRItime}
\centering
\begin{tabular}{ccc}
\toprule
\textbf{Method} & \textbf{Total Time (s)} & \textbf{Time/Scenario (s)} \\
\midrule
Perturbation & 223.40 & 2.23 \\
Dual LP & 89.04 & 0.89 \\
\bottomrule
\end{tabular}%
\end{table}

On the frozen 100-scenario benchmark, the dual formulation reduces wall-clock time from 223.40 s to 89.04 s while also removing the step-size tuning burden from the MRI calculation. Measured on a per-scenario basis, this is a reduction from 2.23 s/scenario to 0.89 s/scenario, i.e., speedup of about \(2.51\times\). In practical large-scale accreditation studies, the Monte Carlo sample size \(N_s\) may exceed \(10^5\), where even a simple linear extrapolation already implies a runtime of about 62 h for perturbation versus about 25 h for the dual method on the same hardware.

Another benefit of the proposed accreditation method is that it is independent of step size, whereas step size is critical in perturbation-based analysis. If the step size is too large, adjacent piecewise segments may overlap; if it is too small, the resulting MRI values may still be inaccurate \cite{qiancredit}. At present, step-size selection relies heavily on engineering judgment, but Corollary \ref{coro2} provides a direct way to calculate MRI for energy storage resources.

\subsection{Guidance for System Operators on Storage Modeling}
Under the MRI capacity accreditation framework, the discount factor (rMRI) discriminates to different energy storage types because of the dispatch rule. In reality, the charging and discharging orders are not decided by the left duration under the reliability dispatch rule but decided by the energy market clearing results. However, the market behavior of storage is hard to model, especially under long-term planning problems.\\ 
\indent Except for the reliability dispatch rule, the system operator might set the dispatch priority based on other factors, such as the storage operation cost or historical performance. For example, assuming a dispatch rule \emph{always} follows the simple charging and discharging priority: storage 1 $>$ storage 2 $>$ storage 3 $>$ storage 4. Based on the same Monte Carlo experiment data, the reliability contribution of all energy storage is the average reduction of EUE to 492.2 MWh, which is less than the reliability dispatch method.

On the other hand, let the MRI of a thermal generator with zero outage rate be the perfect MRI. Under the reliability dispatch model, the perfect MRI is 1.24 MWh/MW, increasing to 1.27 MWh/MW under the simple dispatch model. These two factors make the storage MRI smaller while $\text{MRI}_\text{perfect}$ larger, together resulting in a smaller capacity discount factor (rMRI) for energy storage under a simple dispatch rule.

If QC is defined as the power capacity of storage, Table \ref{rMRI1} compares the impact of dispatch rules on rMRI, where the energy storage is given less credited capacity under the simple dispatch rule.  

\begin{table}[!t] 
\caption{The rMRI calculation results based on different dispatch rules (QC = $\Bar{x}$)}
\label{rMRI1}
\centering
\begin{tabular}{ccc} 
\toprule
\textbf{Group} & \multicolumn{2}{c}{\textbf{rMRI} ($\text{MRI}/\text{MRI}_{\text{perfect}}$)}                            \\
 \#                         & Reliability Dispatch & Simple Dispatch 
\\ \midrule
1&                 0 &                      0 \\
2&                  0.65 &                      0.64 \\
3&                   0.98 &                      0.92 \\
4&                  1.00 &                     0.94\\
\bottomrule
\end{tabular}
\end{table}

To better reflect the reliability contribution of storage with short duration, other definitions of QC might be used by system operators. For example, the QC in Table \ref{rMRI2} is defined as a combination of storage power and energy capacity, but a similar capacity accreditation shrinking phenomenon is also witnessed after changing to a simple dispatch rule.

\begin{table}[!t] 
\caption{The rMRI calculation results based on different dispatch rules (QC = 0.5 $\Bar{x}$ +  0.5 $\Bar{S}$ / duration)}
\label{rMRI2}
\centering
\begin{tabular}{ccc} 
\toprule
\textbf{Group} & \multicolumn{2}{c}{\textbf{rMRI} ($\text{MRI}/\text{MRI}_{\text{perfect}}$)}                            \\
 \#                         & Reliability Dispatch & Simple Dispatch 
\\ \midrule
1&                 0.16 &                      0.16 \\
2&                  0.84 &                      0.82 \\
3&                   1.00 &                      0.94 \\
4&                  1.00 &                     0.94\\
\bottomrule
\end{tabular}
\end{table}
An interesting finding from the above results is that, under the simple dispatch rule, storage 1 consistently has the highest priority for charging and discharging. However, its capacity accreditation, after being adjusted by the rMRI, does not increase and may even decrease. In fact, the role of $\text{MRI}_\text{perfect}$ is critical in the MRI-based accreditation scheme. From a game theory perspective, maximizing the overall reliability impact of the entire resource type may be more effective for achieving higher credit capacity than targeting the reliability impact of individual resources. This principle could also serve as a guide for system operators when modeling other types of resources.

In addition, the current case is locally insensitive to the tested RTE range, which is not obvious a priori. When the charging efficiency is reduced from 0.95 to 0.85, the system-level results are unchanged at the reported precision: annual EUE after storage remains 481.2 MWh, the perfect MRI remains 1.24 MWh/MW, and the total power-based QMRIC remains 350.38 MW. The main reason is that deficiency events are relatively rare and are typically preceded by enough surplus hours for recharge. In addition, only the charging leg carries an efficiency penalty, while discharging efficiency is taken to be close to 1. As a result, moderate charging losses do not materially change the marginal adequacy contribution in this particular system.

\subsection{Suggestions for Policymakers on Reliability Criteria}
The ICR plays an important role in establishing the nominal system capacity for both the demand curve and MRI calculation in FCM \cite{zhao2017constructing}. In practice, since the capacity market operates as a forward market, system operators first forecast load profiles for the future capacity delivery year and then determine the corresponding ICR based on the reliability standard and resource mix required by policymakers.

The traditional resource adequacy standard, based on the LOLE of 1-day-in-10 years, is increasingly insufficient for addressing emerging risks. These risks arise from greater variability and uncertainty caused by the evolving resource mix and increasing demand levels. In response, the North American Electric Reliability Corporation (NERC) has recently been evaluating a multi-criteria approach that incorporates metrics such as loss of load hours (LOLH) and normalized expected unserved energy (NEUE) \cite{nerc2023ltra,nae2023report}. Similarly, planners outside North America facing comparable resource adequacy challenges have adopted different strategies. For example, the Australian National Energy Market Reliability Panel uses an NEUE threshold of less than or equal to 0.002 \% per year as its reliability standard \cite{epri2022report}.
\begin{figure*}[t]
\centering
\includegraphics[scale=0.56]{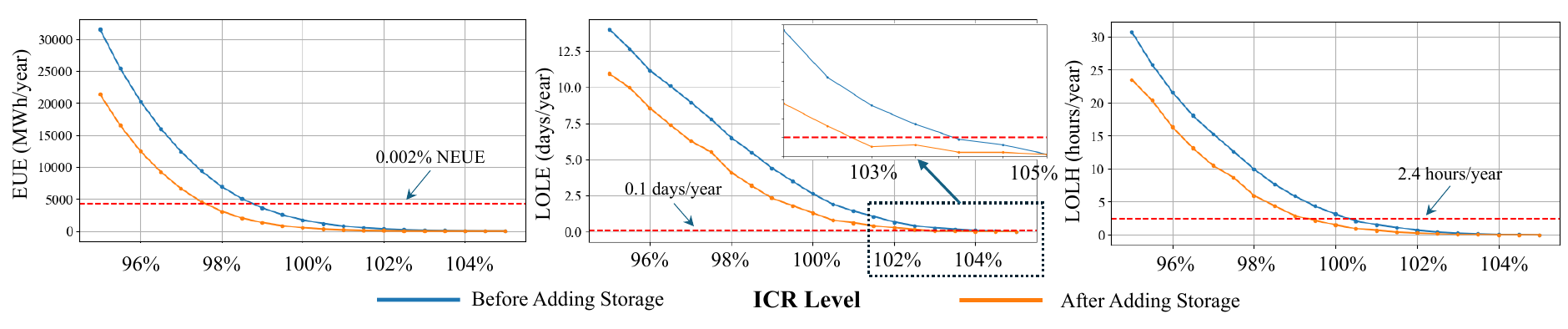}
  \caption{System reliability under different metrics and ICR levels}
  \label{icrmetrics}
\end{figure*}
For the same resource mix, Fig. \ref{icrmetrics} illustrates the values of various reliability criteria across different ICR levels, with the baseline in our test case derived from historical installed capacity and load data from California. It is clear that all the reliability criteria are non-increasing as more capacity is installed, but each criterion measures different risk aspects, with preferences toward the duration, magnitude, or frequency of load loss. The red dotted lines in Fig. \ref{icrmetrics} represent commonly used reliability standards: 0.1 days/year for LOLE, 2.4 hours/year for LOLH, and 0.002 \% energy/year for EUE. While the baseline ICR meets the LOLH and EUE requirements with a margin, an additional 3\% capacity is needed to satisfy the LOLE standard.

Policymakers' choice of reliability standards significantly impacts energy storage capacity accreditation, which is highly sensitive to ICR. Figure \ref{rmriicr} illustrates the rMRI values under different QC definitions and varying ICR levels. As ICR increases, both the system's unserved energy hours and magnitude decrease, indicating that the reliability contribution of energy storage becomes less constrained by its energy capacity. This reduction in the MRI gap compared to a perfect thermal resource leads to an increase in storage rMRI as ICR continues to rise.

\begin{figure}[!t]
\centering
\includegraphics[width=\columnwidth]{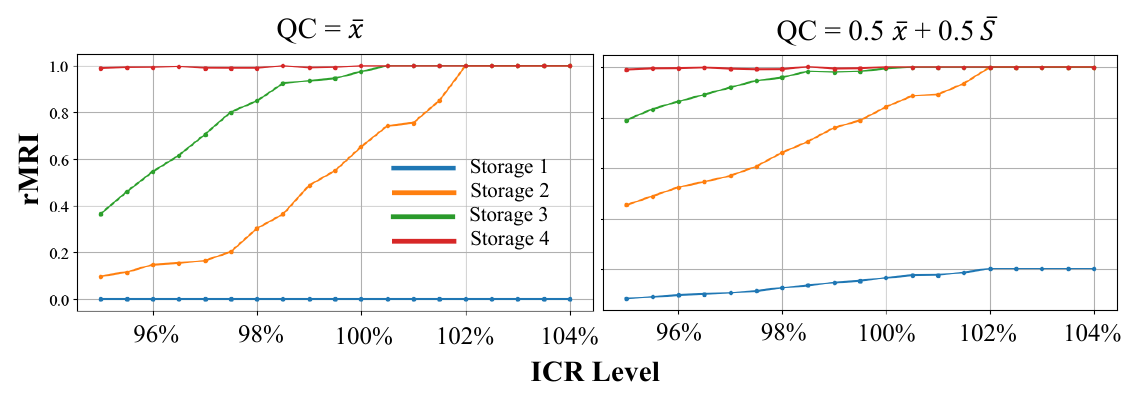}
  \caption{The storage rMRI under different ICR levels}
  \label{rmriicr}
\end{figure}

In addition to selecting more robust reliability metrics, policymakers should study core storage-modeling assumptions such as charging efficiency conventions and temporal resolution. The California case study shows that efficiency losses may have little effect when scarcity events are rare and recharge opportunities are abundant, whereas time resolution can materially change accredited capacity when some effective storage durations are shorter than one hour.

\section{Conclusion} \label{sec7}
This paper presents a tractable MRI framework for energy-storage capacity accreditation by representing the adequacy-oriented storage dispatch rule as a linear program and deriving valid dual variables directly from that formulation. First, we clarify the theoretical scope of the LP reformulation by proving its equivalence to the industry-used storage dispatch algorithm. Second, our method enables direct and efficient calculation of the MRI, bypassing computationally intensive perturbation analysis. Third, we provide explicit numerical validation against the perturbation benchmark and examine sensitivity results in a modified California case study.

The findings deepen the understanding of energy storage's marginal reliability impact and the factors affecting capacity accreditation, including qualified capacity definitions and dispatch rules, making these insights particularly relevant to ongoing marginal accreditation reforms in U.S. ISOs. As the energy landscape continues to evolve with more renewable and storage resources, this methodology provides a foundation for improving market mechanisms, ensuring fair and precise accreditation, and supporting policy decisions.

Future research will focus on enhancing storage aggregation methods with attention to simulation step size, incorporating additional factors such as battery degradation, and exploring the relationship between the MRI of storage and renewable resources to further improve accuracy and applicability.

\section*{Disclaimer}
The views expressed in this paper are solely those of the authors and do not necessarily represent those of ISO New England, PJM Interconnection L.L.C. or its Board of Managers.

\bibliographystyle{IEEEtran}
\bibliography{ref}
\end{document}